\documentclass[a4paper,cleveref, autoref, thm-restate]{lipics-v2019}
\hideLIPIcs{}
\usepackage{charter}
\usepackage{comment}
\excludecomment{shortver}
\includecomment{longver}
\usepackage{wrapfig}
\usepackage{tikz}
\usetikzlibrary{positioning}
\usetikzlibrary{arrows}
\usetikzlibrary{decorations.pathmorphing}
\usetikzlibrary{decorations.markings}

\tikzset{
  treenode/.style = {align=center, inner sep=0pt, text centered,
    font=\sffamily},
  arn_t/.style = {treenode, circle, white, font=\sffamily\bfseries, fill=SeaGreen,
   text width=0.7em},% arbre rouge noir, noeud noir
  arn_n/.style = {treenode, circle, white, font=\sffamily\bfseries, draw=black,
    fill=black, text width=0.7em},% arbre rouge noir, noeud noir
  arn_r/.style = {treenode, circle, red, fill=IndianRed,
    text width=0.7em, very thick},% arbre rouge noir, noeud rouge
  arn_x/.style = {treenode, rectangle, draw=black,
    minimum width=0.5em, minimum height=0.5em}% arbre rouge noir, nil
        photon/.style={decorate, decoration={snake}, draw=red},
    electron/.style={draw=blue, postaction={decorate},
        decoration={markings,mark=at position .55 with {\arrow[draw=blue]{>}}}},
    gluon/.style={draw=DodgerBlue},
    m1/.style={draw=DarkOrchid, thick},
    m2/.style={draw=Chocolate, thick},
    m3/.style={draw=YellowGreen, thick},
    leaf/.style = {draw=IndianRed, circle},
    root/.style = {draw=SeaGreen, circle},
    internal/.style = {draw=Black, circle},
}

\usepackage{latexsym,amssymb,amsmath,mathtools,eulervm,xspace}

\renewcommand{\leq}{\leqslant}
\renewcommand{\geq}{\geqslant}

\renewcommand{\bar}[1]{\ensuremath{\overline{#1}}\xspace}

\newcommand{\NO}{\textsc{No}\xspace}

\newcommand{\NP}{\ensuremath{\mathsf{NP}}\xspace}

\newcommand{\NPC}{\ensuremath{\mathsf{NP}}-complete\xspace}

\newcommand{\EE}{\ensuremath{\mathcal E}\xspace}

\newcommand{\II}{\ensuremath{\mathcal I}\xspace}
\newcommand{\JJ}{\ensuremath{\mathcal J}\xspace}

\newcommand{\VV}{\ensuremath{\mathcal V}\xspace}

\newcommand{\defproblem}[3]{
  \vspace{1mm}
\begin{center}
\noindent\fbox{

  \begin{minipage}{.9\linewidth}
  \begin{tabular*}{\linewidth}{@{\extracolsep{\fill}}lr} \textsc{#1}   \\ \end{tabular*}
  {\bf{Input:}} #2  \\
  {\bf{Question:}} #3
  \end{minipage}

  }
\end{center}
  \vspace{1mm}
}

\newcommand{\AG}{\ensuremath{A}\xspace}
\newcommand{\AGS}{\ensuremath{V}\xspace}
\newcommand{\RES}{\ensuremath{R}\xspace}

\newcommand{\GEFA}{\EE-\textsc{GEFA}\xspace}
\newcommand{\LPA}{\EE-\textsc{LPA}\xspace}

\newcommand{\LSAT}{\textsc{LSAT}\xspace}
\newcommand{\threecol}{\textsc{3-Coloring}\xspace}

\newcommand{\tw}{{\mathbf{tw}}}

\title{On Fair Division with Binary Valuations Respecting Social Networks} %TODO Please add

\titlerunning{On Fair Division with Binary Valuations Respecting Social Networks} %TODO optional, please use if title is longer than one line

\author{Neeldhara Misra}{Indian Institute of Technology, Gandhinagar \and \url{https://people.iitgn.ac.in/~neeldhara} }{neeldhara.m@iitgn.ac.in}{https://orcid.org/0000-0003-1727-5388}{The author would like to acknowledge the SERB ECR Grant for their support of this work}%TODO mandatory, please use full name; only 1 author per \author macro; first two parameters are mandatory, other parameters can be empty. Please provide at least the name of the affiliation and the country. The full address is optional

\author{Debanuj Nayak}{Indian Institute of Technology, Gandhinagar}{debanuj.nayak@alumni.iitgn.ac.in}{}{}

\authorrunning{N. Misra and D. Nayak} %TODO mandatory. First: Use abbreviated first/middle names. Second (only in severe cases): Use first author plus 'et al.'

\Copyright{John Q. Public and Joan R. Public} %TODO mandatory, please use full first names. LIPIcs license is "CC-BY";  http://creativecommons.org/licenses/by/3.0/

 %TODO mandatory; please add comma-separated list of keywords

\category{} %optional, e.g. invited paper

\relatedversion{} %optional, e.g. full version hosted on arXiv, HAL, or other respository/website
\supplement{}%optional, e.g. related research data, source code, ... hosted on a repository like zenodo, figshare, GitHub, ...

\nolinenumbers %uncomment to disable line numbering

\begin{document}

\maketitle

\begin{abstract}
We study the computational complexity of finding fair allocations of indivisible goods in the setting where a social network on the agents is given. Notions of fairness in this context are ``localized'', that is, agents are only concerned about the bundles allocated to their neighbors, rather than every other agent in the system. We comprehensively address the computational complexity of finding locally envy-free and Pareto efficient allocations in the setting where the agents have binary valuations for the goods and the underlying social network is modeled by an undirected graph. We study the problem in the framework of parameterized complexity. 

We show that the problem is computationally intractable even in fairly restricted scenarios, for instance, even when the underlying graph is a path. We show NP-hardness for settings where the graph has only two distinct valuations among the agents. We demonstrate W-hardness with respect to the number of goods or the size of the vertex cover of the underlying graph. We also consider notions of proportionality that respect the structure of the underlying graph and show that two natural versions of this notion have different complexities: allocating according to the notion that accounts for locality to the greatest degree turns out to be computationally intractable, while for other notions, the allocation problem can be modeled as a structured ILP which can be solved efficiently.
\end{abstract}

\section{Introduction}
\label{sec:typesetting-summary}

% !TeX root = ./adt.tex

The problem of fairly allocating resources among a set of agents with (possibly distinct) interests in said resources is a fundamental problem with important and varied practical applications. We focus on the problem of allocating indivisible items: in this setting, we have $n$ agents and $m$ resources, and every agent expresses their utilities for the resources, either as a ranking over the resources or by specifying a valuation function. The goal is to determine an \emph{allocation} of the items to the agents that respects some notion of ``fairness'' and 
``efficiency''. We use the term \emph{bundle} to refer to the set of items that an agent receives in an allocation. 

Envy-freeness is one of the most widely used notions of fairness. Given an allocation, an agent envies another if it perceives the bundle of the other agent to be more valuable than her own. An allocation is \emph{envy-free} if no agent envies another. Note that the trivial allocation that leaves every agent empty-handed is always envy-free. Therefore, one is typically interested in fair allocations that also satisfy some criteria of economic \emph{efficiency}, such as completeness (every good should be allocated to some agent), non-wastefulness (no agent receives a piece of cake that is worth nothing to her and worth something to another agent), or Pareto-efficiency (there is no other feasible agreement that would make at least one agent strictly better off while not making any of the others worse off). We remark here that just as there are trivial allocations that are fair, it is also possible to trivially achieve efficiency if we had no fairness considerations involved: for instance, the allocation that gives all goods to a single agent is Pareto-efficient assuming that the agent has a strictly monotonic utility function over the items. 

The question of finding allocations that respect fairness and efficiency demands simultaneously is non-trivial: in particular, such allocations may not exist (if there are two agents and one good, and both agents have positive utility for this single resource), and can be computationally hard to find (for instance, the problem of finding a complete envy-free allocation between even two agents who hold identical valuations over $m$ goods is equivalent to the \textsc{Partition} problem). %Therefore, much attention has been devoted to identifying reasonable notions of fairness and efficiency that satisfy existence guarantees and are computationally tractable. For example, the notion of envy-freeness has been relaxed to \emph{envy-freeness up to one good} (EF1)~\cite{Budish2011}, where an agent may hypothetically remove a good from the bundle of another agent when comparing it to its own --- it turns out that with this relaxation, there is no need to trade efficiency for fairness: for additive valuations, an allocation that maximizes the Nash social welfare --- defined to be the geometric mean of the agents' valuations --- is both fair (EF1) and Pareto efficient and can also be computed efficiently when the utilities are bounded~\cite{BarmanKV18}. Other approaches to relaxed fairness criteria include eliminating envy using money~\cite{BrustleDNSV20}, and allowing for sharing of resources~\cite{abs-1908-01669}, or achieving ``local'' envy-freeness when the agents are related by a social network structure~\cite{Bredereck0N18,BeynierCGLMW18}.% 

The focus of this work is the notion of local envy-freeness. In this setting, the agents are related by a graph, which might be thought of as modeling a social network over the agents, and we explore notions of fairness that account for the structure of this network. For instance, the notion of envy is now restricted: it only manifests between agents who are friends in the network. This is a compelling model of fairness, since agents are likely to not envy agents about whom they have little or no information. We note that the problem of fair division respecting a social network generalizes the classical notion, which can be captured by considering a complete graph on the agents. Thus, the problem of finding allocations that are ``locally fair'' is a generalization of the classical allocation problem.

\subsection{Related Work}

The model of local envy-freeness has been proposed and considered in several recent lines of work. Some of the earliest considerations for incorporating a graph structure on the agents were made in the context of the \emph{cake-cutting} problem, which is the closely related setting of allocating a divisible resource among agents~\cite{AbebeKP17,BeiQZ17}. Abebe, Kleinberg, and Parkes~\cite{AbebeKP17} consider both directed and undirected graphs and focus on characterizing the structure of graphs that admit algorithms with certain bounds. They also consider the issue of the \emph{price of envy-freeness} in this setting, which compares the total utility of an optimal allocation to the best utility of an allocation that is envy-free. Bei, Qiao, and Zhang~\cite{BeiQZ17}, on the other hand, propose a moving-knife algorithm that outputs an envy-free allocation on trees and an algorithm for computing a proportional allocation on descendant graphs.

We now turn to the literature in the context of indivisible items. Beynier et al~\cite{BeynierCGLMW18} study the fair division problem in the setting of ``house allocation'': here agents have (strict) preferences over items, and each agent must receive exactly one item. An agent envies another in this setting if she prefers the item received by the other agent over her own. In the case of a complete network, for an allocation to be envy-free, each agent must get her top object, and this assignment is automatically Pareto-efficient as well. This motivates the setting of local envy-freeness with respect to a graph on the agents. The authors consider the case when the underlying graph is undirected, and they also consider a variant of the problem where agents themselves can be located on the network by the central authority. These problems turn out to be computationally intractable even on very simple graph structures. 

Bredereck, Kaczmarczyk, and Niedermeier~\cite{Bredereck0N18} consider the problem of graph-based envy-freeness in the context of \emph{directed} graphs and for various classes of valuations: including binary, identical, additive, and even valuations that are both identical and binary. They also consider the complexity of the allocation problem in the framework of parameterized complexity~\footnote{The terminology relevant to this framework is introduced in the next section.}. Somewhat surprisingly, it turns out that finding complete envy-free allocations in the setting of a graph is NP-hard even when the valuations are binary \emph{and} identical. Note that in this setting, every agent in every strongly connected component must get the same number of items: thus, the allocation problem is trivial for directed graphs that are strongly connected, but NP-hard for general directed graphs. Also, it turns out that for general binary preferences, the problem of finding a complete envy-free allocation is NP-hard even when the graph is strongly connected. The problem is also tractable for DAGs: indeed, allocating all resources to a single source agent (corresponding to a vertex with no incoming arcs) is both complete and locally envy-free since nobody can envy a source agent, and empty-handed agents have no envy for each other.  

More recently, Eiben et al.~\cite{EibenGHO20} consider the problem of finding locally envy-free allocations and envy-free allocations that are additionally proportional in the setting of directed graphs in the framework of parameterized complexity, and specifically considering parameters such as treewidth, cliquewidth, and vertex cover --- all of these reflect the structure of the underlying network. It turns out that the problem of finding fair and efficient allocations is tractable for networks that have bounded values for these parameters with some additional assumptions that bound the number of item types or the size of the largest bundle received by an agent. The authors also show hardness results in both the parameterized and classical settings. For instance, the authors show that finding a locally envy-free allocation is NP-hard even when the underlying network is a star, but we note that this is in the setting of general utilities. 

The work of Bredereck et al~\cite{abs-2005-04907,BKKN} demonstrates that the problem of finding fair and efficient allocations in various settings (including graph-based constraints) is fixed-parameter tractable in the combined parameter ``number of agents'' and ``number of item types'' for general utilities. In contrast, our work here focuses on smaller parameters for the special case of binary utilities.

%Aziz et al.~\cite{ABCGL18} also consider social constraints in the context of indivisible items, but their notion of graph-based envy-freeness requires not only that the agents do not envy their neighbors, but that additionally, with respect to their non-neighbors, any potential envy can be eliminated by a hypothetical re-allocation of all items outside of the agent's bundle. Our work, and indeed much of the literature on local envy-freeness, does not impose the additional requirement in the context of non-neighbors.%

In~\cite{FD2}, Chevaleyre, Endriss, and Maudet consider \emph{distributed} mechanisms for allocating indivisible goods, in which agents can locally agree on deals to exchange some of the goods in their possession. This study focuses on convergence properties for such distribution mechanisms both in the context of the classical setting and the setting involving social constraints coming from an underlying undirected graph. Here, the notions of fairness localized according to the graph, and the network also constraints the exchanges that can take place --- agents can engage in an exchange only if they are friends in the network. There are also some lines of work that suggest eliminating envy by some mechanism for \emph{hiding} information~\cite{HosseiniSVWX20}. 
    
\subsection{Our Contributions}

% Our contributions here specifically address a research direction that was suggested by~\cite{Bredereck0N18}, where the underlying network was modelled by directed graphs. Often, our results are in contrast with reasonable analogues in the world of directed graphs: for instance, we show that computing EEF allocations with respect to an undirected graph turns out to be NP-hard even when the underlying graph is a path (Theorem~\ref{thm:nphpaths} and Corollary~\ref{thm:nphpathscor}), while such allocations are known to be computable in polynomial time for DAGs.

Our focus in this paper is on the setting when agents have binary valuations over the goods and the underlying social network is modeled by an undirected graph. Our focus is on exploring the computational complexity of finding locally envy-free allocations that are also Pareto efficient (EEF) in the framework of parameterized complexity, building most closely on the works of~\cite{EibenGHO20,BliemBN16,Bredereck0N18}. 

\paragraph*{Bounded Agent Types.} We begin by noting that the setting of undirected graphs can be significantly different from their directed counterparts: indeed, recall that finding a complete and locally envy-free allocation was NP-hard for even identical binary valuations for directed graphs, but the analogous question is easily seen to be tractable for undirected graphs (indeed, observe that the notions of strong connectivity and connectivity coincide). This motivates the question of whether the problem of finding locally EEF allocations is easier for undirected graphs with a bounded number of agent types. We answer this question in the negative by showing that the problem of determining locally envy-free allocations is NP-hard even when there are only two distinct binary valuations among the agents by a reduction from a graph separation problem called \textsc{Cutting $\ell$ Vertices} (Theorem~\ref{thm:2agenttypes}). 

\paragraph*{Sparse and Dense Graphs.} In contrast with the result for DAGs, we show that finding locally envy-free allocations that are Pareto efficient (EEF) is NP-hard even when the underlying graph is a path (Theorem~\ref{thm:nphpaths} and Corollary~\ref{thm:nphpathscor}). Although Beynier et al~\cite{BeynierCGLMW18} also show hardness results for very sparse graphs, we note that our methods are significantly different since the models for the valuations are different and additionally, the allocations we seek need not give every agent exactly one item. Moving away from sparsity, we recall that finding complete envy-free allocations for binary valuations is known to be NP-hard even for complete graphs~\cite{HosseiniSVWX20,AzizGMW15}, which justifies the need for using additional parameters\footnote{The algorithm referred to is XP in the cliquewidth of the underlying graph, the number of agent types and item types.} in the XP\footnote{XP is the class of parameterized problems that can be solved in time $n^{f(k)}$ for some computable function $f$. } algorithm for finding locally envy-free allocations shown by~\cite[Theorem 10]{EibenGHO20}.

\sloppypar
\paragraph*{Structural Parameters I: Treewidth and Cliquewidth} Informally speaking, the parameters treewidth and cliquewidth of graphs quantitatively capture the sparsity and density of the graph by measuring their ``likeness'' to trees and complete graphs. The results we have already for sparse and dense graphs demonstrate that these parameters being bounded alone is not enough to obtain tractable algorithms. On the other hand, the results of~\cite{EibenGHO20} imply that the problem of finding complete and locally envy-free allocations admits XP algorithms when parameterized by either the treewidth or cliquewidth of the underlying graph jointly with the number of item types and agent types. Since their model allows for bidirectional edges, these results apply to the setting of undirected graphs as well. We note that the algorithms described in~\cite{EibenGHO20} focus on complete allocations, but can be adapted to account for Pareto efficiency as well. 

\paragraph*{Structural Parameters II: Vertex Cover and Twin Cover} In the setting of directed graphs and general utilities, we note that the problem of finding a complete and locally envy-free allocation is NP-hard even when the underlying graph is a star. In particular, this demonstrates hardness on graphs with a constant-sized vertex cover\footnote{A vertex cover of a graph is a subset of vertices that contains at least one endpoint of every edge. A graph with a bounded vertex cover also has bounded treewidth.}. It is not clear if this is the case for undirected graphs and binary utilities. We show that the problem of finding locally EEF allocations is W[1]-hard when parameterized by the vertex cover number (Theorem~\ref{thm:vcwhard}). We remark that a stronger hardness result can be observed for the closely related parameter of twin cover\footnote{A twin cover of a graph is a subset of vertices $S$ such that $G \setminus S$ is a disjoint union of cliques, and further, every pair of vertices $u,v$ in any clique of $G \setminus S$ are ``twins'', that is, $N[v] = N[u]$.} --- indeed, the known NP-hardness of finding envy-free allocations for binary valuations on complete graphs~\cite{HosseiniSVWX20,AzizGMW15} implies hardness for graphs that have a twin cover of size zero.

% We note that the fact that finding EF  and para-NP-hard for a closely related parameter called the twin cover\footnote{A twin cover of a graph is a subset of vertices $S$ such that $G \setminus S$ is a disjoint union of cliques, and further, every pair of vertices $u,v$ in any clique of $G \setminus S$ are ``twins'', that is, $N[v] = N[u]$.}. In particular, we demonstrate hardness on the very restricted class of graphs that consists of a disjoint union of cliques and global vertex (Theorem~\ref{thm:tchard}).   

\paragraph*{Few Resources or Agents.} We also consider the cases where the number of goods or the number of agents are relatively small. When considering these parameters, the work of Bliem et al~\cite{BliemBN16} shows that the computation of EEF allocations is FPT when parameterized by the number of goods or the number of agents for additive 0/1 valuations. In contrast, we show that finding EEF allocations respecting the structure of an underlying undirected graph is $W[1]$-hard when parameterized by the number of goods (Theorem~\ref{thm:whard}). On the other hand, the FPT algorithm when parameterized by the number of agents can be extended to account for the graph constraints (noted in Observation~1).

% We also establish that computing EEF allocations is FPT when parameterized by the treewidth of the underlying graph and the number of goods (Theorem~\ref{thm:twfpt}). It is natural to consider this combined parameter since we encounter intractability when either is considered separately. Indeed, the treewidth of a graph is a measure of its resemblance to a tree, and in particular, paths have treewidth one: thus the previous result implies the para-NP-hardness of the problem when parameterized by treewidth.

\paragraph*{Other Notions of Fairness.} Finally, we also consider notions of proportionality in the context of graphs --- we refer to these as local and quasi-global proportionality concepts, representing the extent to which the definitions account for the underlying graph. We demonstrate that computing a locally proportional allocation is NP-hard (Theorem~\ref{thm:nphlprop}), while computing a proportional allocation that is quasi-global is tractable (Theorem~\ref{thm:gpropptime}). Notions of local proportionality have been proposed and studied in several of the papers that were summarized in the previous section.

\section{Preliminaries}

We use standard terminology from graph theory and fair division. Unless mentioned otherwise, the graphs we consider are simple and undirected. For a graph $G = (V, E)$, consisting of a set $V$ of vertices and a set $E$ of edges, by $N(v)$ we denote the \emph{neighborhood} of vertex $v \in V$ , i.e., the set $W \subset V$ of vertices such that for each vertex $w \in W$ there exists an edge $e = (v,w) \in E$. The \emph{closed neighborhood} of a vertex $v$ is $N(v) \cup \{v\}$ and is denoted $N[v]$. The \emph{degree} of a vertex $v$, denoted $d(v)$, is $|N(v)|$. A \emph{clique} is a subset of vertices which are pairwise adjacent. An \emph{independent set} is a subset of vertices, no two of which are adjacent. For $X \subseteq V$, the \emph{induced subgraph} $G[X]$ denotes the subgraph whose vertex set is $X$ and the edge set consists of all edges whose both end points are in $X$. 
% We denote the treewidth of $G$ by $\tw(G)$; this is a parameter measuring the resemblance of $G$ to a tree. \shortonly{For the formal definition, we refer the reader to the full version of this manuscript.~\cite{full}} \longonly{The formal definition is as follows.}

\begin{longver}
    \begin{definition}
        Let $G$ be a graph.  A {\em tree-decomposition} of a graph $G$ is a pair
        $\mathbb{T} = (T,(B_\alpha)_{\alpha\in V(T)})$, where $\mathbb{T}$ is a rooted tree, such that
        \begin{itemize}
        \item $\cup_{t\in V(T)}{B_t}=V(G)$,
        \item for every edge $xy\in E(G)$ there is a $t\in V(T)$ such that  $\{x,y\}\subseteq B_{t}$, and
        \item for every  vertex $v\in V(G)$ the subgraph of $T$ induced by the set  $\{t\mid v\in B_{t}\}$ is connected.
        \end{itemize}

        The {\em width} of a tree decomposition is $\max_{t\in V(\mathbb{T})} |B_t| -1$ and the {\em treewidth} of $G$
        is the  minimum width over all tree decompositions of $G$ and is denoted by $\tw(G)$.
    \end{definition}

For completeness, we also define here the notion of a \emph{nice tree decomposition with introduce edge nodes}, as this is what we will work with in due course. We note that for a given tree decomposition can be modified in linear time to fulfill the above constraints; moreover, the number of nodes in such a tree decomposition of width~$w$ is $O(w \cdot n)$~\cite{Kloks94}.

    \begin{definition}
    A tree decomposition $\mathbb{T}=(T,(B_\alpha)_{\alpha\in V(T)})$ is a \emph{nice tree decomposition with introduce edge nodes} if the following conditions hold.
    \begin{enumerate}
    \item The tree~$T$ is rooted and binary.
    \item For all edges in $E(G)$ there is exactly one \emph{introduce edge node} in $\mathbb{T}$, where an introduce edge node is a node~$\alpha$ in the tree
      decomposition~$\mathbb{T}$ of~$G$ labeled with an edge~$\{u,v\}\in
      E(G)$ with~$u,v\in B_\alpha$ that has exactly one child
      node~$\alpha'$; furthermore~$B_\alpha=B_{\alpha'}$.
    \item Each node
      $\alpha\in V(T)$ is of one of the following types:
      \begin{itemize}
      \item introduce edge node;
      \item \emph{leaf node}: $\alpha$ is a leaf of $T$ and $B_\alpha=\emptyset$;
      \item \emph{introduce vertex node}: $\alpha$ is an inner node of $T$ with
        exactly one child node $\beta\in V(T)$; furthermore~$B_\beta\subseteq
        B_\alpha$ and~$|B_\alpha\backslash B_\beta|=1$;
      \item \emph{forget node}: $\alpha$ is an inner node of $T$ with exactly one
        child node $\beta\in V(T)$; furthermore~$B_\alpha\subseteq B_\beta$
        and~$|B_\beta\backslash B_\alpha|=1$;
      \item \emph{join node}: $\alpha$ is an inner node of $T$ with exactly two
        child nodes $\beta,\gamma\in V(T)$; furthermore~$B_\alpha = B_\beta =
        B_\gamma$.
      \end{itemize}
    \end{enumerate}
\end{definition}
\end{longver}

% The {\em adhesion} of an edge $\{t,t'\}\in E({\mathbb T})$ is $\vert X_t\cap X_{t'}\vert$. The  {\em adhesion} of a tree decomposition is the maximum adhesion of an edge in ${\mathbb T}$. For any $t\in V(\Tb)$, we use $\gamma(t)$ to denote the set $\bigcup_{t'\in {\sf descendent}(t)} X_{t'}$
% (Here $t\in {\sf descendent}(t)$).
An instance of fair division for indivisible goods consists of $n$ \emph{agents} $\AG = \{1, \ldots, n\}$ and $m$ \emph{goods} (also called \emph{items} or \emph{resources}), $\RES = \{o_1, \ldots, o_m\}$. Further, we are also given \emph{valuations} (also called \emph{preference functions} or \emph{utilities}) $\nu_\ell: 2^\RES \rightarrow \mathbb{Z} $ for every agent $\ell \in \AG$. We will assume throughout that the valuation functions are \emph{additive}, i.e., for each agent $\ell \in \AG$ and any set of goods $S \subseteq \RES$, $\nu_\ell(S) := \sum_{o \in S} \nu_\ell(\{o\})$. A \emph{$0/1$ valuation} is a function that takes values in $\{0,1\}$, while valuations are said to be \emph{identical} if every agent has the same preference function. In the context of $0/1$ valuations, we say that an agent \emph{values} or \emph{approves} a good if her utility for the good is $1$. We will use $\VV$ to denote the valuations of the agents $\AG$ over $\RES$. When considering fair division in the context of social networks, we are also given an undirected graph $G$ over the agents $\AG$.

% Without loss of generality, we adopt the convention that the valuation $\nu$ is given by the associated function $\nu^\prime: \RES \rightarrow \mathbb{Z} \cup \{0\}$ and that $\nu^\prime(\emptyset) = 0$.

Every subset $S \subseteq \RES$ is called a \emph{bundle}. An allocation is a function $\pi: \AG \rightarrow 2^{\RES}$ mapping each agent to the bundle she receives, such that $\pi(i) \cap \pi(j) = \emptyset$ when $i \neq j$ because the items cannot be shared. When $\bigcup_{a \in \AG} \pi(a) = \RES,$ the allocation $\pi$ is said to be \emph{complete}, otherwise it is \emph{partial}. An allocation is \emph{non-wasteful} if every good is allocated to an agent that assigns positive utility to it.

An allocation $\pi^\prime$ \emph{dominates} $\pi$ if for all $\ell \in \AG$ it holds that $\nu_\ell(\pi(\ell))) \leq \nu_\ell(\pi^\prime(\ell))$ and for some $a_j \in A$ it holds that $\nu_{a_j}(\pi(a_j))) < \nu_{a_j}(\pi^\prime(a_j))$. An allocation $\pi$ is \emph{Pareto-efficient} if there exists no allocation $\pi^\prime$ that dominates $\pi$. In the case of $0/1$ preferences, we note that an allocation is Pareto-efficient if and only if it is complete and non-wasteful, assuming that each resource provides a value of 1 to at least one agent.

Given an instance of fair division $(\AG,\RES, G = (\AG,E),\VV)$ as described above, we now introduce the following fairness notions:

\begin{itemize}
    \item \textbf{Graph Envy-Freeness (GEF).} We call allocation $\pi$ graph-envy-free if for each pair of (distinct) agents $i,j \in \AG$ such that $j \in N(i)$, it holds that $\nu_i(\pi(i)) \geq \nu_i (\pi(j))$.
    \item \textbf{Global Proportionality (GP).} We say that an allocation $\pi$ achieves global proportionality if for each agent $\ell \in \AG$, $\nu_i(\pi(i)) \geq \frac{1}{n} \nu_i(\RES)$.
    \item \textbf{Quasi-Global Proportionality (QP).} We say that an allocation $\pi$ achieves quasi-global proportionality if for each agent $\ell \in \AG$, $\nu_i(\pi(i)) \geq \frac{1}{d(\ell) + 1} \nu_i(\RES)$.
    \item \textbf{Local Proportionality (LP).} We say that an allocation $\pi$ achieves local proportionality if for each agent $\ell \in \AG$, $\nu_i(\pi(i)) \geq \frac{1}{d(\ell) + 1} \sum_{j \in N[i]} \nu_i(\pi(j))$.
\end{itemize}

Note that the graph versions of variants of envy-freeness (such as EF1 or EFX) can be defined analogously in a straightforward manner. It is easy to see that any graph envy-free allocation is also locally proportional and that if the underlying graph is complete, then local proportionality coincides with the standard notion of proportionality. For the problems we consider, we are typically given an instance of fair division on a graph, and the goal is to determine if there exists an allocation that satisfies some notion of fairness and efficiency. For instance, consider the following problems:

\defproblem{Graph Envy-Free Allocation (\GEFA)}{An instance of fair division on a graph\\ $(\AG,\RES,G = (\AG,E),\VV)$.}{Does there exist an envy-free, Pareto-efficient allocation?}

\defproblem{Locally Proportional Allocation (\LPA)}{An instance of fair division on a graph\\ $(\AG,\RES,G = (\AG,E),\VV)$.}{Does there exist a Pareto-efficient allocation that achieves local proportionality?}

For any efficiency concept (X) and fairness notion (Y), the X-YA problem is defined in a similar fashion. Although our questions are posed as decision versions, we note that most of our algorithms can be easily adapted to handle the natural ``search'' version of these problems. We refer the reader to the books~\cite{FD1,FD3} and the article~\cite{FD2} for additional background on fair division.

A problem \emph{parameterized} by $k$ is fixed-parameter tractable if it is solvable in $f(k)|I|^{O(1)}$ time for some computable function $f$ and the input size $|I|$ according to the problem’s encoding. Informally, $W$-hard problems are presumably not fixed-parameter tractable. The problem of finding a clique on at least $k$ vertices is $W[1]$-hard when parameterized by $k$. We call a problem para-NP-hard if it is NP-hard even for a constant value of the parameter. For a comprehensive introduction to the paradigm of paramterized complexity and algorithms, we refer the reader to the book~\cite{CyganFKLMPPS15}.

% \nmtodo{Define LSAT, Clique, etc.}

\section{Envy-Freeness}

\subsection{NP-hardness for two agent types}

% !TeX root = ./adt.tex

In this section, we show that finding \GEFA{} allocations is NP-hard even in the setting of near-identical binary valuations: in particular, when all agents have one of two possible utilities over the items. Note that in the setting of identical binary valuations when the graph $G$ is connected, it is easy to see that all agents must value all goods without loss of generality, and that desirable allocations are the ones that allocates the same number of goods to each agent, where the goods themselves may be arbitrarily chosen. Indeed, it is clear that an allocation with equal bundle sizes is \GEFA{}. On the other hand, consider a \GEFA{} allocation that does not allocate bundles of equal size to all agents. Let $a_i$ and $a_j$ be two agents that receive bundles of different size. We can always find two adjacent agents on a path from $a_i$ to $a_j$ who have received bundles of different sizes, contradicting envy-freeness.  

%Consider a path $P$ in the graph $G$ from $a_i$ to $a_j$ --- it is easy to see that by considering the bundle sizes of the agents on the path $P$, we will find two adjacent agents who have received bundles of different sizes, contradicting the assumption that the allocation was envy-free.%

We now show that even a slightly more general situation is computationally intractable --- in particular, if all agents have one of two valuations over the goods, the problem of identifying \GEFA{} allocations is NP-hard.

\begin{theorem}
    \label{thm:2agenttypes}
    The \GEFA problem is \NPC{} even when there are two agent types, and further, agents have 0/1 valuations over the goods.
\end{theorem}

\begin{proof}
We reduce from the \textsc{Cutting $\ell$ vertices} problem, where we are given a graph $G = (V,E)$, integers $\ell$ and $k$, and the question is if there exists a partition of the vertex set $V$ into $X \cup S \cup Y$ such that $|X| = \ell$, $|S| \leq k$ and there is no edge between $X$ and $Y$. It follows from~\cite{BuiJ92} that this problem is NP-hard. We now describe the instance of \GEFA{}. 
We introduce $2n+1$ agents $w_1,\ldots,w_n$, $u_1,\ldots,u_n$ and $s$, where $n = |V|$. We call $s$ the trigger agent and we also refer to the $w_i$'s and $u_i$'s as greedy and happy agents, respectively. The graph structure on the agents is defined as follows. For every edge $e = (v_i,v_j) \in E$, we introduce the edges $(w_i,u_j)$ and $(w_j,u_i)$. We also make the trigger agent adjacent to every greedy agent. We also add the edges $(w_i,u_i)$ for all $1 \leq i \leq n$. We now turn to the items. We introduce $\ell$ items called the \emph{coveted} items: $\{g_1,\ldots,g_\ell\}$. Every agent has a utility of one for these items. We introduce $(n-\ell+1)$ items $\{p_1,\ldots,p_{n-\ell+1}\}$ called the $w$-type items and $\ell+k$ items $\{q_1,\ldots,q_{\ell+k}\}$ called the $u$-type items. The trigger agent and the greedy agents have utility one for all $w$-type items and the happy agents have utility one for all $u$-type items. This completes the description of the reduction. %Note that the happy agents have an uniform valuation over all goods, while the trigger agent and the greedy agents have an uniform valuation over all goods.

We first argue the forward direction. Given a subset of at most $k$ vertices $S$, we allocate all the coveted goods to the $\ell$ greedy agents corresponding to vertices of $X$, and the $w$-type items to the remaining greedy agents and the trigger agent. Also, allocate the $k+\ell$ $u$-type items to happy agents corresponding to vertices in $X \cup S$. Observe that this allocation is both locally envy free and Pareto efficient, since the only empty-handed agents are the happy agents corresponding to vertices of $Y$, but the only agents they potentially envy are agents in $X$ (since these are the agents who got the coveted goods) --- however, recall that there are no edges between $X$ and $Y$. 

In the reverse direction, let a locally envy free, Pareto efficient allocation be given. To begin with, note that the trigger agent ensures that no coveted good is allocated to a happy agent: indeed, if the agent $u_i$ has a coveted good, then this generates envy in $w_i$, who must be given either a coveted good or a $w$-type good. This in turn makes the trigger agent envious, who must also be given either a coveted good or a $w$-type good. At this point, all the remaining $(n-1)$ greedy agents are envious (of the trigger agent), but there aren't enough goods to account for all of them: thus any allocation that does not assign all coveted goods to greedy agents is not locally envy free. Let $X$ be the subset of vertices corresponding to the $\ell$ greedy agents who were allocated the coveted goods, and let $S^\star$ be the subset of vertices corresponding to $(k + \ell)$ happy agents who were given $u$-type goods.  Note that $X \subseteq S^\star$ and $|S^\star \setminus X| \leq k$. It is easy to check that if we let $S := S^\star \setminus X$ and $Y := V \setminus (X \cup S)$, then $(X \cup S \cup Y)$ is a partition of the desired kind. Indeed, if not, there is an edge between a vertex in $X$ and a vertex in $Y$, but since the happy agents corresponding to vertices in $Y$ are (by definition) empty-handed and the agents corresponding to vertices of $X$ have been assigned a coveted good, this edge would violate local envy-freeness, a contradiction. This concludes the argument. 
\end{proof}
% More formally, if there is no good that is not valued by any agent and the underlying graph $G$ is connected, an allocation is \GEFA{} for identical binary valuations if and only if all bundle sizes are equal.%

\subsection{W-hardness parameterized by goods}

In this section, we demonstrate the hardness of finding \GEFA allocations even when the number of goods is bounded by showing that the problem is $W[1]$-hard when parameterized by the number of goods.

\begin{theorem}
    \label{thm:whard}
    The \textsc{\GEFA} problem is $W[1]$-hard when parameterized by the number of goods, even when agents have 0/1 valuations over the goods.
    \end{theorem}

    We describe a reduction from the $W[1]$-hard problem \textsc{Clique}, given a graph $G$ and an integer $k$, does there exist a clique on $k$ vertices in $G$. Let $\II = (G, k)$ be an instance of clique, where $G = (V,E)$ and further, $V = \{v_1, \ldots, v_n\}$ and $E = \{e_1, \ldots, e_m\}$. We assume, without loss of generality, that $m \geq {k \choose 2}$, since we can always return a trivial \NO{}-instance when this is not the case. We begin by describing the construction of the reduced instance $\JJ_\II := (\AG,\RES, H = (\AG,F), \VV)$. We define the set of goods $\RES$ as follows:

    $$ \RES = \left\{q_1, \ldots, q_\ell, p_1, \ldots, p_k, d_1, \ldots, d_{\ell+1} \right\}, $$

    where $\ell = {k \choose 2}$. For ease of discussion, we call the first $\ell$ goods \emph{popular} and the next $k$ goods \emph{specialized}. The remaining are \emph{dummy} items. We now define the set of agents as $\AG = V \cup E \cup S \cup W$, where:
    $$S := \{s_1,\ldots,s_{\ell + 1}\} \mbox{ and } W := \{w_{ij} ~|~ i \in [n], j \in [\ell + 1] \}. $$

    We indulge in a mild abuse of notation and use $v_i$ to refer to both an element of $V$ from the clique instance and an agent of $A$ in the reduced instance (similarly for edges). The edges of $H$ are as follows:

    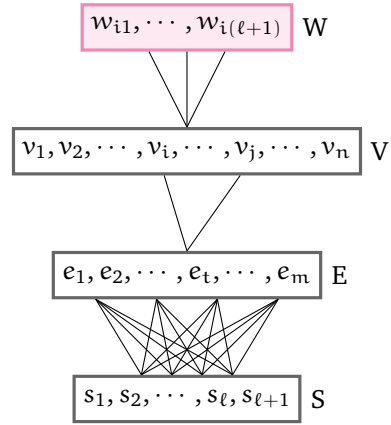
\begin{wrapfigure}{r}{0.4\textwidth}
    \centering
    \begin{tikzpicture}[
    W/.style={rectangle, draw=magenta!60, fill=magenta!10, very thick, minimum size=6mm},
    E/.style={rectangle, draw=black!60, fill=white!5, very thick, minimum size=6mm},
    V/.style={rectangle, draw=black!60, fill=white!5, very thick, minimum size=6mm},
    S/.style={rectangle, draw=black!60, fill=white!5, very thick, minimum size=6mm},
    ]
    %Nodes
    \node[W] (W) [label=right:W]{$w_{i1}, \cdots, w_{i(\ell+1)}$};
    \node[V] (V) [below=of W,label=right:V] {$v_1 , v_2 , \cdots, v_i, \cdots, v_j, \cdots, v_n $};
    \node[E] (E) [below = of V,label=right:E] {$e_1 , e_2 , \cdots, e_t, \cdots, e_m $};
    \node[S] (S) [below = of E,label=right:S] {$s_1 , s_2 , \cdots, s_\ell, s_{\ell+1} $};

    %Lines
    \begin{scope}
        \foreach \i in {-1,0,1}{%
        \draw[-] (V.north) -- ([xshift=\i * 0.5 cm]W.south) ;}
        \foreach \i in {-1.5,3.5}{%
        \draw[-] (E.north) -- ([xshift=\i * 0.2 cm]V.south) ;}
        \foreach \i in {-6,-2,2,6}{%
        \foreach \j in {-3,-1,1,3}{
        \draw[-] ([xshift=\j * 0.2 cm]S.north) -- ([xshift=\i * 0.2 cm]E.south) ;} }
    \end{scope}
    \end{tikzpicture}
    \caption{A sketch of the reduced instance based on an instance $G = (V,E), k$ of \textsc{Clique}. Recall that $\ell$ denotes ${k \choose 2}$, and only some vertices of $W$ are shown for clarity. The shaded vertices induce a complete subgraph. The edge $e_t = (v_i, v_j)$ is adjacent only to $v_i$ and $v_j$ among vertices in $V$.}
\end{wrapfigure}

    \begin{itemize}
    \item $e = (u,v) \in E$ is adjacent to all vertices of $S$ and $u,v$.
    \item $v_i \in V$ is adjacent to all vertices $w_{ij}$ for $j \in [\ell + 1]$.
    \item For each $1 \leq i \leq n$, $H[\cup_{j = 1}^{\ell+1}w_{ij}]$ induces a clique.
    \end{itemize}

    The preferences of the agents are as follows:

    \begin{itemize}
    \item All agents have an utility of $1$ for the popular goods.
    \item All agents in $V$ have an utility of $1$ for the specialized goods.
    \item The agent $s_i \in S$ has an utility of $1$ for $d_i$, for all $i \in [\ell + 1]$.
    \end{itemize}

    This completes the construction of the instance. Note that the number of goods is a function of $k$ alone. We now turn to the argument for equivalence,  a clique $X \subseteq V$, of size $k$ exists in $G$ \emph{iff}, there is an GEF allocation for the instance constructed $\JJ_\II := (\AG,\RES, H = (\AG,F), \VV)$ 

    \begin{proof}
    In the forward direction, let $X \subseteq V$ be a clique in $G$ and let $Y := G[X] \subseteq G$. Consider now the following allocation $\pi$. We let each agent corresponding to $Y$ receive one popular item, each agent corresponding to $X$ receive one specialized item, and finally allocate the item $d_i$ to $s_i$ for  all $i \in [\ell + 1]$. 
    
    % \shortonly{It is straightforward to verify that the allocation $\pi$ is Pareto-efficient and envy-free with respect to $H$.}
 
\begin{longver}
    \begin{claim}
    The allocation $\pi$ is Pareto-efficient and envy-free with respect to $H$.
    \end{claim}

    \begin{proof} For Pareto-efficiency, it suffices to argue that $\pi$ is complete and non-wasteful. This is evident from the definition of $\pi$ --- indeed, all the goods were allocated, and further, these specialized goods were allocated only to the vertex agents who like them and any dummy good was allocated to the unique agent from $S$ who had a non-zero valuation for it. The popular goods are liked by everyone, and in particular the agents who received them.

    We now address envy-freeness. Consider $v \in X$:  the corresponding agent receives one of the specialized goods and its adjacent edges get at most one popular good each, thus there is no envy for $v$. If $v \notin X$, the corresponding agent doesn't get anything but neither do any of its adjacent edges, so again, there is no envy. Now consider $e \in E$. These agents are connected only to agents in $V \cup S$, who receive either specialized or dummy goods --- since agents in $E$ do not value these goods, there is no envy.
    % If both endpoints of $e$ belong to $X$, then the agent corresponding to $e$ gets a popular good and nobody in its neighborhood gets a popular good, thus she has no envy. Any of the other edges, although empty-handed, are

    The agents in $S$ each have received a dummy good and all the edges connected to them have received only one popular good, so they don't envy any of their neighbors. Finally, the agents in $W$ have not been allocated anything. They are connected to vertices from $V$ who are either also empty-handed or have specialized goods that agents from $W$ don't value. Thus no envy for agents in $W$. This concludes the argument.
    \end{proof}
\end{longver}
    This concludes our description of a fair and complete allocation strategy given a clique in $G$. We now turn to the reverse direction, where we are given an allocation $\pi$ that is Pareto-efficient and envy-free with respect to $H$. It is useful to make the following observation about $\pi$ to begin with.

    \begin{claim}
        Let $H$ be defined as above, and let $\pi$ be an allocation that is Pareto-efficient and envy-free with respect to $H$. Then, any popular good is assigned by $\pi$ to an agent from $E$. Further, no agent in $E$ can receive more than one popular good in the allocation $\pi$.
    \end{claim}
\begin{longver}
    \begin{proof}
    Let us assume that there exists a popular good $q$ that is not assigned to an agent from $E$. This means that, since $\pi$ is complete, $q$ is assigned to an agent from $V$, $S$, or $W$.
    \begin{itemize}
        \item If $q$ was assigned to an agent $v_i$ from $V$, then all the vertices $w_{ij}$ for $j \in [\ell+1]$ will envy $v_i$. Note that there are $\ell+1$ such agents, each of which only value the popular goods, and the total number of (remaining) popular goods is at most $\ell-1$, the allocation $\pi$ cannot account for all the envy generated in this situation.

        \item If $q$ was assigned to an agent $w_{ij}$ from $W$, then all the other agents in $W$ belonging to the clique induced by $N(v_i)$ will envy $w_{ij}$. There are $\ell$ such agents, each of which only value the popular goods, and the total number of (remaining) popular goods is at most $\ell-1$, and again, the allocation $\pi$ cannot account for all the envy generated in this situation.

        \item If $q$ was assigned to an agent $s_i$ from $S$, then all the agents in $E$ will envy $q$. Since $|E| \geq \ell$, and all agents in $E$ only value the popular goods, the argument is the same as in the previous case.
    \end{itemize}

    From the discussion above, we conclude that $\pi$ allocates the popular goods to agents in $E$. Now suppose there exists an agent $e$ in $E$ who received more than one popular good. Since $e$ is adjacent to all the $s_i$'s in $S$, the presence of more than one good for $e$ triggers envy in all the $s_i$ agents. There are $(\ell + 1)$ many agents in $S$, each of whom must now be assigned \emph{two} items that they value. However, apart from the unique dummy good that each $s_i$ values, the only goods that they value are the popular goods, of which only $(\ell - 2)$ remain --- therefore, we fall short of accounting for the envy of $(\ell + 1)$ agents. This shows that no agent in $E$ can receive more than one popular good.
    \end{proof}
\end{longver}

    Since $\pi$ is non-wasteful, the specialized goods must be distributed among agents corresponding to $V$. The following is easy to see.
    
    % It is easy to see the following as well (c.f.~\cite{full}):

    \begin{claim}
        Let $H$ be defined as above, and let $\pi$ be an allocation that is Pareto-efficient and envy-free with respect to $H$. No agent in $V$ can receive more than one specialized good in the allocation $\pi$.
    \end{claim}

\begin{longver}
    \begin{proof}
    We know that $\ell = {k \choose 2}$ many edge agents have each received a popular good. These edges combined span at least $k$ many agents in $V$ (recall that $G$ is a simple graph). Each of those $k$ agents suffer from envy as they are adjacent to edges who possess popular goods that the vertices also like. To satisfy this envy, each of these $k$ vertices must be assigned a specialized good --- indeed, the popular goods, the only other possibility, are all taken. The claim follows from the fact that there are only $k$ specialized goods.
    \end{proof}
\end{longver}
    Let $X \subseteq V$ be the subset of $k$ agents that receive at least one specialized item and let $Y \subseteq E$ be the subset of $\ell$ agents that receive at least one popular item with respect to $\pi$ . We claim that $G[X]$ is a clique. In particular, we claim that every edge of $Y$ has both its endpoints in $X$. Indeed, suppose not, and let $e \in Y$ be an edge with at least one endpoint (say $v$) outside $X$. Then, $v$ envies $e$, which contradicts our assumption about $\pi$ being envy-free with respect to $H$.
    \end{proof}

\subsection{W-hardness parameterized by vertex cover}
% !TeX root = ./adt.tex

Recall that a vertex cover of a graph $G = (V,E)$ is a subset $S \subseteq V$ such that $G \setminus S$ is an independent set (i.e, for any pair of vertices $u,v \in G \setminus S$, $(u,v) \notin E$). In the setting of directed graphs with arbitrary utilities, finding \GEFA{} allocations is NP-hard even for graphs that have a constant-sized vertex cover. Here, we show that in the setting of binary utilities, finding a \GEFA{} allocation is $W[1]$-hard when parameterized by the vertex cover of the underlying graph. 

\begin{theorem}
    \label{thm:vcwhard}
    The \textsc{\GEFA} problem is $W[1]$-hard when parameterized by the vertex cover of the underlying graph, even when agents have 0/1 valuations over the goods.
\end{theorem}

\begin{proof}
    We describe a reduction from the $W[1]$-hard problem \textsc{Clique}, given a graph $G$ and an integer $k$, does there exist a clique on $k$ vertices in $G$. Let $\II = (G, k)$ be an instance of clique, where $G = (V,E)$ and further, $V = \{v_1, \ldots, v_n\}$ and $E = \{e_1, \ldots, e_m\}$. We begin by describing the construction of the reduced instance. 

    To begin with, we introduce $k$ ``key'' agents $u_1, \ldots, u_k$, and for every pair $1 \leq i, j \leq k$, we introduce $m$ agents $x_{ij}^1, \ldots, x_{ij}^m$ --- these agents are all adjacent to $u_i$ and $u_j$. We call these agents the $(i,j)$ ``guard'' agents. We also introduce $n-k$ ``residual'' agents and impose a complete bipartite graph between the key and residual agents, that is, every residual agent is adjacent to every key agent. Note that this graph has a vertex cover of size $k$ given by the key agents. 

    We now turn to the items. We introduce a ``core'' item $g_i$ for each vertex $v_i \in V$. The key agents and residual agents value all core items. We also introduce, for every pair $1 \leq i, j \leq k$, $m-1$ dummy items valued by all $(i,j)$ guard agents. Finally, the guard agent $x_{ij}^\ell$ values all the core items \emph{except} the ones corresponding to the endpoints of the edge $e_\ell$. This completes the description of the reduction.
    
    In the forward direction, given a clique $S \subseteq V$, we assign core items corresponding to the clique to the key agents, and the remaining core items to the residual agents. If the key agents $a_i$ and $a_j$ are assigned goods $g_p$ and $g_q$, then let $\ell_{ij}$ denote the edge $(v_p,v_q)$. We say that the index $\ell_{ij}$ is ``special'' for the $(i,j)$ guard agents. Now, for every guard agent $x_{ij}^\ell$ except the special index, we assign the corresponding dummy item. Notice that when $\ell = \ell_{ij}$, the guard $x_{ij}^\ell$ does \emph{not} envy either of the key agents it is adjacent to, since it has an utility of one for all goods except for $g_p$ and $g_q$. 

    In the reverse direction, it is easy to check that due to the limited number of core goods, no locally envy-free allocation gives a guard agent a core good (indeed, this can be seen to induce a chain of envy that leaves us with having to allocate $n-1$ goods among $n$ agents which is not feasible). It is also easy to check that all the key and residual agents get exactly one core good each. We claim that the vertices corresponding to the core goods assigned to the key agents, say $S$, form a clique. Indeed, suppose not, and let $v_p$ (the vertex corresponding to the good assigned to key agent $a_i$) and $v_q$ (the vertex corresponding to the good assigned to key agent $a_j$) be two non-adjacent vertices in $S$. Then it follows that all the $(i,j)$ guard agents envy both $a_i$ and $a_j$, but given that no core goods are allocated to the guard agents and there are only $(m-1)$ dummy goods that can be assigned among these guards, it is inevitable that at least one of the guards will envy one of these key agents, contradicting our assumption of local envy-freeness. 
\end{proof}

\subsection{NP-hardness on paths}
% !TeX root = ./adt.tex

To show the hardness of \GEFA even when the underlying graph is a path, we reduce from a variant of SAT called \textsc{Linear SAT} (abbreviated \textsc{LSAT}). In an \LSAT{} instance, each clause has at most three literals, and further the literals of the formula can be sorted such that every clause corresponds to at most three consecutive literals in the sorted list, and each clause shares at most
one of its literals with another clause, in which case this literal is extreme in both clauses. The hardness of LSAT was shown in~\cite{Arkin2015}. In fact, by studying the reduced instance, one may assume that a ``hard'' instance of \LSAT{} has the following structure: the first $2q$ clauses have two literals each and are of the following form:

$$A_i = \{s_i, \ell_i\}, B_i = \{\ell_i, t_i\}; 1 \leq i \leq q,$$

where $s_i, \ell_i,$ and $t_i$ denote literals, while the remaining $p$ clauses have three literals each and are mutually disjoint from each other as well as the first $2q$ clauses. For ease of description, we will assume that the \LSAT{} formula that we reduce from has this particular structure. We are now ready to describe our reduction --- in the interest of simplicity, our proof is designed to address the case when the graph is a disjoint union of paths, although it is easy to ``stitch'' these components into a single, longer path, as we will explain later.

\begin{theorem}
    \label{thm:nphpaths}
    The \GEFA problem is \NPC{} even when the graph induced by the agents is a disjoint union of paths, and further, agents have 0/1 valuations over the goods.
\end{theorem}

%\begin{proof}
Membership in \NP{} is straightforward to check. We focus here on the reduction demonstrating hardness. Let $\phi$ be an instance of \LSAT{} over variables $\hat{X} := \{x_1, \ldots, x_n\}$ and clauses:

$$C := \{A_1, B_1, \ldots, A_q, B_q, C_1, \ldots, C_p\},$$

as described above. We refer to the first $2q$ clauses as the \emph{coupled} clauses and the remaining as \emph{isolated} clauses. We now turn to the construction of the reduced instance $\JJ_\phi := (\AGS,\RES, H, \VV)$. We define the set of goods $\RES$ as $\RES_{\hat{X}} \cup \RES_C$, where:
$$\RES_{\hat{X}} = \left\{y_1, \ldots, y_n, x_1, \ldots, x_n,  \bar{x}_1, \ldots, \bar{x}_n,   \right\},$$
and:
$$\RES_C = \left\{ g_1, \ldots, g_q, d_1, \ldots, d_p \right\}. $$

The set of agents $\AGS$ is given by $X~\cup~C~\cup~Y~\cup~G~\cup~D$, where $C$ is denoted in the same way as in the \LSAT{} instance, and further:

\begin{longver}
$$X = \{X_1, \ldots, X_n\}, Y = \{Y_1, \ldots, Y_n\}, G = \{G_1, \ldots, G_q\} \mbox{ and } D = \{D_1, \ldots, D_p\}.$$
\end{longver}
\begin{shortver}
$$X = \{X_1, \ldots, X_n\}, Y = \{Y_1, \ldots, Y_n\},$$
$$G = \{G_1, \ldots, G_q\} \mbox{ and } D = \{D_1, \ldots, D_p\}.$$
\end{shortver}

We now simultaneously describe the structure of the graph $H$ and the preferences of the agents.

\begin{figure}
    \centering
    \begin{tikzpicture}[
    roundnode/.style={circle, draw=green!60, fill=green!5, very thick, minimum size=8mm},
    bignode/.style={circle, draw=red!60, fill=red!5, very thick, minimum size=10mm}]

    %Nodes
    \node[roundnode]      (dummy)          [label=below:$y_i$] {\large{$Y_i$}};
    \node[bignode]        (var)       [right=of dummy, label=below:$y_i$ $x_i$ $\bar{x}_i$] {\large{$X_i$}};
    \node[roundnode]      (dummyiso)          [right= of var, label=below:$d_i$] {\large{$D_i$}};
    \node[bignode]        (clauseiso)       [right=of dummyiso, label=below:$d_i$ $\jmath_1$ $\jmath_2$ $\jmath_3$] {\large{$C_i$}};
    \node[roundnode]      (dummycoup)          [right= of clauseiso, label=below:$g_i$] {\large{$G_i$}};
    \node[bignode]        (clausecoup1)       [right=of dummycoup, label=below:$g_i$ $s_i$ $t_i$ $\ell_i$] {\large{$A_i$}};
    \node[bignode]        (clausecoup2)       [right=of clausecoup1, label=below:$s_i$ $t_i$] {\large{$B_i$}};
    %Lines
    \draw[-] (dummy.east) -- (var.west);
    \draw[-] (dummyiso.east) -- (clauseiso.west);
    \draw[-] (dummycoup.east) -- (clausecoup1.west);
    \draw[-] (clausecoup1.east) -- (clausecoup2.west);
    \end{tikzpicture}
    \caption{A schematic of the reduced instance in the proof of Theorem~\ref{thm:nphpaths}, which is a disjoint union of paths. The vertex labels denote the agents while the goods they approve are indicated just below. Here, $\jmath_1, \jmath_2,$ and$\jmath_3$ denote the items corresponding to the literals that belong to the clause $C_i$.} 
\end{figure}
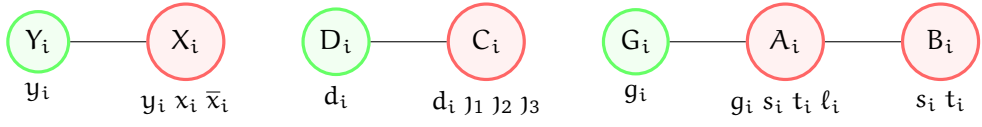
\textbf{Assignment Gadgets.} For each $1 \leq i \leq n$, add an edge between $X_i$ and $Y_i$. The agent $X_i$ values $\{x_i, \bar{x}_i, y_i\}$, while $Y_i$ values the good $y_i$ (and nothing else).

\textbf{Isolated Clause Gadgets.} For each $1 \leq i \leq p$, we add an edge between agents $C_i$ and $D_i$. The agent $C_i$ values the literal $\ell$ if and only if $\ell \in C_i$ along with $d_i$, while $D_i$ values the good $d_i$ (and nothing else).

\textbf{Coupled Clause Gadgets.} For each $1 \leq i \leq q$, we add an edge between agents $A_i$ and $B_i$, and also an edge between $G_i$ and $A_i$. The agent $G_i$ values the good $g_i$ (and nothing else). Agents $A_i$ and $B_i$ value, respectively, the goods $\{g_i, s_i, t_i, \ell_i\}$ and $\{s_i,t_i\}$.

This completes the description of the construction. 

We now discuss the equivalence. In the forward direction, given a satisfying assignment $\tau$ for $\phi$, we define the allocation $\pi$ as follows:

\begin{itemize}
    \item For all $1 \leq i \leq n$, $\pi(Y_i) = \{y_i\}$ and:
    \begin{equation*}
        \pi(X_i) =
        \begin{cases}
          x_i & \text{if } \tau(x_i) = 0,\\
          \bar{x}_i & \text{if } \tau(x_i) = 1.
        \end{cases}
    \end{equation*}

    Observe that the literals that evaluate to true under $\tau$ are not assigned to any of the agents from $X$.
    \item For all $1 \leq i \leq p$,  $\pi(D_i) = \{d_i\}$ and $\pi(C_i)$ is the item corresponding to one of the literals of $C_i$ that evaluates to $1$ under $\tau$. If there are multiple literals that evaluate to $1$ under $\tau$, allocate all items corresponding to said clause to the agent $C_i$.
    \item For all $1 \leq i \leq p$,  $\pi(G_i) = \{g_i\}$. Further, if $\tau$ sets $\ell_i$ to true, then $\pi(A_i) = \{\ell_i\}$ and $\pi(B_i) = \emptyset$. However, if $\tau$ sets $\ell_i$ to false, then note that it must be the case that the literals $s_i$ and $t_i$ evaluate to true under $\pi$, and we assign $\pi(A_i) = \{s_i\}$ and $\pi(B_i) = \{t_i\}$.
    \item  If both $s_i$ and $t_i$ evaluate to true under $\tau$, then they can be distributed one each among the agents $A_i$ and $B_i$, and if exactly one of them is set to true then the good corresponding to the true literal can always be given to $B_i$.
\end{itemize}

It is easy to verify that $\pi$ is well-defined: this follows from the fact that the clauses are almost disjoint. It is straightforward to check that the allocation is complete and non-wasteful (and hence Pareto-efficient). We now observe that the allocation is also envy-free: indeed, the agents involved in paths that are edges get one item each while the agents involved in paths of length three either get one item each, or in the case when $B_i$ is empty-handed, $A_i$ receives a good for which $B_i$ has zero utility (and in this case, both $A_i$ and $G_i$ receive one good each), thus there is no envy.

In the reverse direction, let $\pi$ be an envy-free, Pareto-efficient allocation. Since $\pi$ is non-wasteful, note that the item $y_i$ is allocated by $\pi$ to either $Y_i$ or $X_i$. However, if $y_i \in \pi(X_i)$, then the agent $Y_i$ envies $X_i$, but since she approves of no item other than $y_i$, it is not possible for $\pi$ to resolve this envy. Therefore, combined with the fact that $\pi$ is non-wasteful, we conclude that $\pi(Y_i) = \{y_i\}$ for all $1 \leq i \leq n$. A similar argument establishes that $\pi(D_i) = \{d_i\}$ for all $1 \leq i \leq p$ and $\pi(G_i) = \{g_i\}$ for all $1 \leq i \leq q$.

Since $\pi(Y_i) = \{y_i\}$, we have that $\pi(X_i) \neq \emptyset$. By the non-wastefulness of $\pi$, we also know that $\pi(X_i) \subseteq \{x_i,\bar{x}_i\}$. We use this to define an assignment $\tau$ as follows:

\begin{equation*}
    \tau(x_i) =
    \begin{cases}
      1 & \text{if } \pi(X_i) = \{\bar{x}_i\},\\
      0 & \text{if } \pi(X_i) = \{x_i\},\\
      1 & \text{otherwise}.
    \end{cases}
\end{equation*}

It turns out that if $\pi(X_i) = \{x_i, \bar{x}_i\}$, then the setting of $\tau$ for $x_i$ is immaterial, and we set it to $1$ as suggested above as a matter of convention. We now argue that $\tau$ is a satisfying assignment. Suppose not, and in particular, consider a clause that is not satisfied by $\tau$. We address the cases of isolated and coupled clauses separately.

Let $C_i$ be an isolated clause not satisfied by $\tau$. Since $C_i$ approves $d_i$, $\pi(D_i) = d_i$, and $D_i \in N(C_i)$, we know that $\pi(C_i) \neq \emptyset$. In particular, this implies that one of the literals belonging to the clause $C_i$ was allocated to the agent $C_i$. As an example, suppose this literal was $\bar{x}_k$. This implies that $\bar{x}_k \notin \pi(X_k)$, which in turn implies that $\pi(X_k) = \{x_k\}$ and that $\tau(x_k) = 0$, contradicting our assumption that $C_i$ is not satisfied by $\tau$. A symmetric argument holds for the case when the literal in question was a positive literal as opposed to a negated one.

Now, let $A_i$ be a coupled clause not satisfied by $\tau$. If $\ell_i \in \pi(A_i)$, then by an argument similar to the previous case, we can conclude that $\tau(\ell_i) = 1$, which leads to a contradiction. Therefore, assume that $\ell_i \notin \pi(A_i)$. Since $A_i$ approves $g_i$, $\pi(G_i) = g_i$, and $G_i \in N(A_i)$, we know that $\pi(A_i) \neq \emptyset$. Therefore, it must be the case that at least one of $s_i$ or $t_i$ belongs to $\pi(A_i)$. This again implies that $\tau(s_i) = 1$ or $\tau(t_i) = 1$, which leads to the desired contradiction. The same argument works if we started with the coupled clause $B_i$ instead of $A_i$. This discussion concludes the proof.
%\end{proof}

We remark that it is possible to combine the connected components in the reduced instance above by simply introducing ``dummy connector agents'' that each value a corresponding dummy item and nothing else. All the arguments made above will work in exactly the same fashion since nobody would have reason to envy these newly introduced agents, and vice versa. Thus, we have the following corollary.

% From the discussion above, we conclude the NP-hardness of \GEFA{} even on paths and in the setting of 0/1 valuations, summarized in the statement below.

\begin{corollary}
    \label{thm:nphpathscor}
    The \GEFA problem is \NPC{} even when the graph induced by the agents is a path, and further, agents have 0/1 valuations over the goods.
\end{corollary}

\section{Proportionality for Graphs}

\subsection{Local Proportionality: NP-hardness}

% !TeX root = ./adt.tex
\begin{theorem}
    \label{thm:nphlprop}
    The \LPA problem is \NPC{} on undirected graphs, even when all agents have $0/1$ valuations over the resources.
\end{theorem}
% \vspace{-1pt}
\begin{proof}
Membership in NP follows from the fact that an allocation serves as a certificate: indeed, given an allocation, it can be checked whether the allocation is local graph proportional in polynomial time. We show NP-hardness by a reduction from \threecol{}, where we are given a graph $G$ and the question is to determine if there exists a coloring the vertices with 3 colors such that no two adjacent vertices are coloured the same color. Equivalently, we would like to know if the vertex set of $G$ can be partitioned into 3 independent sets.

Let $G = (V,E)$ be an instance of \threecol{}, where $V = \{v_1, \ldots, v_n\}$ and $E = \{e_1, \ldots, e_m\}$. We will now describe the reduced instance, which we denote by $\II_G := (\AG, \RES, H = (\AG,F), \VV)$. We have the following items:

$$ \RES = V \cup E_1 \cup E_2 \cup E_3 \cup \{a_1, a_2, a_3\} \cup \{x_1, x_2, x_3\},$$

where $E_b = \{e_i^b ~|~ i \in [m]\}, b \in [3]$. We refer to the items in $V$ as the \emph{core} goods, the items in $E_b$ as \emph{guard goods of type $b$}. We also introduce the following agents

$$ \AG = F_1 \cup F_2 \cup F_3 \cup \{A_1, A_2, A_3\} \cup \{B_1, B_2, B_3\} \cup C_1 \cup C_2 \cup C_3,$$

where $F_b = \{E_i^b ~|~ i \in [m]\}, b \in [3]$ and $C_b = \{C_i^b ~|~ i \in [n+1] \cup \{0\}\}, b \in [3]$. The preferences of the agents are as follows:

    \begin{itemize}
        \item For each $b \in [3]$, the agent $A_b$ approves all the core goods and the item $a_b$.
        \item For each $b \in [3]$, the agent $B_b$ approves $x_b$.
        \item For an edge $e_i = (v_p,v_q)$ and $b \in [3]$ the agent $E_i^b$ approves the goods $\{e_i^b,v_p,v_q,x_b,a_b\}$.
        \item For each $b \in [3]$, the agents in $C_b$ approve all the core goods and the item $a_b$.
    \end{itemize}

The graph $H$ has the following structure. For each $b \in [3]$, the agent $A_b$ is adjacent to each agent in $F_b$. The agents $C_0^b$ and $B_b$ are adjacent to all agents in $F_b$. For each $i \in [n]$, the agents in $C_b$ induce a clique. We refer the reader to Figure~\ref{fig:lprop} for a schematic depiction of the graph. This completes the description of the construction. We now turn to a discussion of the equivalence of these instances.

In the forward direction, let a coloring of $G$ with three colors $\{1,2,3\}$ be given. We propose an allocation $\pi$ as follows. For any $b \in [3]$, we assign the core goods corresponding to vertices of color $b$ to $A_b$, along with the item $a_i$. For each $b \in [3]$ and $i \in [m]$, we assign the good $e_i^b$ to the agent $E_i^b$. For each $b \in [3]$, assign the good $x_b$ to the agent $B_b$. This describes a complete and non-wasteful allocation of $\RES$ among the agents in $\AG$, and we claim that the allocation is also locally proportional. We establish this by a case analysis. In the following, let $b \in [3]$ be fixed.

\begin{longver}
\begin{itemize}
\item No agent in the neighborhood of $A_b$ receives a good that is valued by $A_b$, and $A_b$ receives at least one good.
\item For any $i \in [m]$, the agent $E_i^b$ values five goods and receives one. Let $e_i = (v_p,v_q)$, without loss of generality, assume that $v_p$ is colored $b$. Indeed, it is possible that neither endpoint of $e_i$ is colored $b$, or that $v_q$ is colored $b$, in which case this argument works in the same way, but importantly, note that it is never the case that both $v_p$ and $v_q$ are colored $b$, since we started with a proper coloring.

Recall that the neighborhood of $E_i^b$ has three agents, whose allocations projected on the goods valued by $E_i^b$ are as follows:

$$ \pi(A_b) = \{a_b, v_p\}, \pi(B_b) = \{x_b\}, \mbox{ and } \pi(C_0^b) = \emptyset.$$

Therefore, the total valuation of $E_i^b$ for the items allocated in its closed neighborhood is four, and $d(E_i^b) + 1 = 4$. This implies that the allocation is locally proportional for $E_i^b$.

\item For any agent $C_i^b \in C_b$, observe that the agents recieve nothing, and no vertex in their closed neighborhood recieves anything valued by them.

\item No agent in the neighborhood of $B_b$ receives a good that is valued by $B_b$, and $B_b$ receives one good.

\end{itemize}
\end{longver}

This concludes the argument in the forward direction. In the reverse, let $\pi$ be a Pareto-efficient allocation that is locally proportional with respect to $H$. We make the following observations.

\begin{claim}For any $b \in [3]$, $\pi(B_b) = \{x_b\}$.
\end{claim}
\begin{longver}
\begin{proof}  These agents don't like any other good, and are connected with the edge agents $E_j^b$. If any of the edge agents receives the $x_b$ goods, then it will be impossible to satisfy the proportionality guarantee for $B_b$. None of the other agents value $x_b$, so the claim follows from the non-wastefulness of $\pi$.
\end{proof}
\end{longver}

\begin{claim}\label{claim:empty} For any $b \in [3]$, and $\pi(C_0^b) = \emptyset$. Further, $\pi(C_i^b) = \emptyset$ for all $i \in [n+1]$.
\end{claim}
\begin{longver}
\begin{proof} If not, then by non-wastefulness, $C_0^b$ receives some subset of the core goods and possibly the item $a_b$. However, since any of these goods are also valued equally by the agents $C_i^b$, for all $1 \leq i \leq n+1$, they must also receive at least one good that they value to satisfy their local proportionality constraints. However, given that between them they only value at most $n$ goods altogether (after excluding the goods allocated to $C_0^b$), it is impossible to extend this allocation to satisfy the constraints of all the agents. The second part of the claim follows by a similar argument.
\end{proof}
\end{longver}

\begin{claim} For any $b \in [3]$, no core good is assigned to an agent from $F_b$.
\end{claim}
\begin{longver}
\begin{proof} If a core good is assigned to any agent in $F_b$, then $\pi(C_0^b) \neq \emptyset$, which contradicts the previous claim.
\end{proof}
\end{longver}

\begin{claim} For any $b \in [3]$, $a_b \in \pi(A_b)$.
\end{claim}
\begin{longver}
\begin{proof} By non-wastefulness, $a_b \in \pi(X)$ for some $X \in \{A_b\} \cup F_b \cup C_b$. By Claim~\ref{claim:empty}, $X \in \{A_b\} \cup F_b$. However, if $a_b$ is assigned to some agent in $F_b$ then $\pi(C_0^b) \neq \emptyset$, which again contradicts Claim~\ref{claim:empty}. Therefore, it follows that  $a_b \in \pi(A_b)$.
\end{proof}
\end{longver}

\begin{claim} For any $b \in [3]$ and $i \in [m]$, $\pi(E_i^b) = \{e_i^b\}$.
\end{claim}
\begin{longver}
\begin{proof}
    This follows from the claims above and the non-wastefulness of $\pi$.
\end{proof}
\end{longver}

Consider the partition $(V_1, V_2, V_3)$ of $V$ given by the core goods assigned to agents $A_1, A_2$ and $A_3$, respectively --- in other words, $V_b$ is defined as the subset of vertices corresponding to the core goods assigned to $A_b$ for each $b \in [3]$. Note that the claims above imply that the core goods can only be assigned to agents $A_1, A_2$ and $A_3$, therefore, the proposed partition accounts for every vertex in $V$. We now claim that $G[V_b]$ is an independent set for all $b \in 3$. Indeed, suppose not. Let $e_i = (v_p,v_q)$ be an edge in $V_b$. Consider the agent $E_i^b$, who has received one item. By the claims above, and our understanding that $A_b$ has received both $v_p$ and $v_q$, we see that the total valuation of $E_i^b$ for the items allocated in its closed neighborhood is five, while its degree is three. This violates the local proportionality constraint for the agent $E_i^b$ and and concludes our argument.
\end{proof}

\begin{figure}[t]
    \begin{center}
    \begin{tikzpicture}[scale=0.65]
        \node[shape=circle,draw=black] (A1) at (0,5) {$A_b$};
        \node[shape=circle,draw=black] (e11) at (-4,0) {$E_1^b$};
        \node[shape=circle,draw=black] (e21) at (-2,0) {$E_2^b$};
        \node[shape=circle,draw=black] (e31) at (0,0) {$E_3^b$};
        \node[shape=circle,draw=black] (e41) at (2,0) {$E_4^b$};
        \node[shape=circle,draw=black] (e51) at (4,0) {$E_5^b$};
        \node[shape=circle,draw=black] (B1) at (3,-3) {$B_b$};
        \node[shape=circle,draw=black] (C1) at (-3,-3) {$C_0^b$};
        \node[shape=circle,draw=black] (c11) at (-4,-6) {$C_1^b$};
        \node[shape=circle,draw=black] (c12) at (-2,-6) {$C_2^b$};
        \node[shape=circle,draw=black] (c13) at (0,-6) {$C_3^b$};
        \node[shape=circle,draw=black] (c14) at (2,-6) {$C_4^b$};

        \path [-](A1) edge node[left] {} (e11);
        \path [-](A1) edge node[left] {} (e21);
        \path [-](A1) edge node[left] {} (e31);
        \path [-](A1) edge node[left] {} (e41);
        \path [-](A1) edge node[right] {} (e51);
        \path [-](B1) edge node[left]  {}(e11);
        \path [-](B1) edge node[left] {} (e21);
        \path [-](B1) edge node[left] {} (e31);
        \path [-](B1) edge node[left] {} (e41);
        \path [-](B1) edge node[right] {} (e51);
        \path [-](C1) edge node[left]  {}(e11);
        \path [-](C1) edge node[left] {} (e21);
        \path [-](C1) edge node[left] {} (e31);
        \path [-](C1) edge node[left] {} (e41);
        \path [-](C1) edge node[right] {} (e51);
        \path [-](C1) edge node[right] {} (c11);
        \path [-](C1) edge node[right] {} (c12);
        \path [-](C1) edge node[right] {} (c13);
        \path [-](C1) edge node[right] {} (c14);

    \end{tikzpicture}
    \end{center}
    \caption{A schematic depiction of the reduction for the problem of finding a locally proportional allocation. Here, $m = 5$ and $n = 3$. The vertices $C_1^b,\ldots,C_4^b$ induce a clique (edges omitted for clarity).}
    \label{fig:lprop}
    \end{figure}
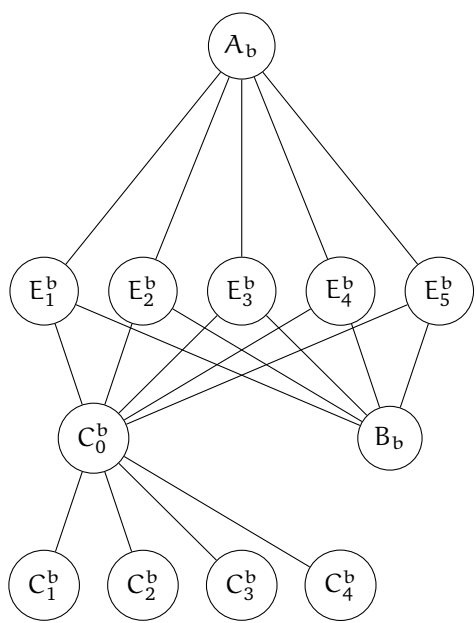

\subsection{Quasi-Global Proportionality: Efficient Algorithms}

To obtain efficient algorithms for finding Pareto-efficient allocations that respect quasi-global proportionality, We model the problem of finding Pareto-efficient allocations respecting quasi-global proportionality using an integer linear program (ILP) with a structured constraint matrix. In particular, it is well-known that if the constraint matrix of an ILP is totally unimodular\footnote{A unimodular Matrix is a square integer matrix having determinant +1 or -1. A totally unimodular matrix is a matrix for which every square non-singular submatrix is unimodular.}, then the corresponding instance can be solved in polynomial time. We turn to an explanation of our encoding.

\begin{theorem}
    \label{thm:gpropptime}
    The problem of finding a Pareto-efficient allocation that is quasi-globally proportional with respect to an underlying undirected graph on the agents can be solved in polynomial time if all agents have $0/1$ valuations.
\end{theorem}

% \begin{longver}
\begin{proof} Let us assume there are $n$ agents and $m$ goods. We will introduce a variable $x_{ij}$ which indicates whether agent $i$ gets good $j$, and $a_{ij}$ indicates whether agent $i$ likes good $j$. These constraints are as follows.
% The goal for the linear program is to maximize the share of one agent by putting appropriate constraints on the goods and rest of the agents.
\begin{itemize}
    % \item Every good must go to a single agent and the agent must like that good. In other words, the allocation is non-wasteful. Goods are not assigned to agents who don't value them. Thus, for all goods $j$, we have the following constraint:
    % \[\sum_{i=1}^{n}{a_{ij}x_{ij}} \leq 1\]

    \item We encode the fact that the allocation defined by $x$ is well-defined by introducing the following constraint for each good $j$:

    \item For each agent $i$, let $s_i$ be the number  of  items  that  have  utility 1 for  agent $i$. For each agent $i$, introduce the following proportionality constraint:
    \[\forall j \sum_{i=1}^n x_{ij} \leq 1 \text{ and } \forall i \sum_{j=1}^{m}{a_{ij}x_{ij}} \geq  \frac{s_i}{(d_i+1)} \]
\end{itemize}

We let the objective function be $\sum_{i=1}^n a_{ij}x_{ij}$. Note that any assignment for which this function achieves a value of $m$ is complete and non-wasteful, and also respects quasi-global proportionality. It is straightforward to verify that the constraint matrix for the ILP described above is totally unimodular for any underlying graph $H$.
\end{proof}

We remark that the problem of assigning goods in a proportional fashion (for any of the notions of proportionality that we have introduced) beyond $0/1$ valuations is NP-hard even when there are only two agents with identical valuations, by a standard reduction from \textsc{Partiton}, with the graph being a singe edge on two agents.
% \end{longver}

\section{Concluding Remarks}

We studied locally EEF allocations in the setting of binary valuations and undirected graphs, and demonstrated that the problem of finding such allocations is computationally intractable for various restricted settings. On the algorithmic front, tools based on dynamic programming and ILP can be used when the instance is structured and has a small number of agent types and item types. 

It is natural to consider notions of fairness and efficiency other than the ones that we explored here. We remark that notions for which allocations can be found efficiently for complete graphs (such as EF1+PO in the setting of binary allocations), the local version is not as interesting, since the allocation that works for a complete graph will work for any graph. For this reason, fairness notions of EF1 or EQx are not relevant to this setting when combined with Pareto efficiency. It would be worth exploring locally fair and efficient allocations for more general valuations, in particular, including chores. In the context of parameterized complexity, a specific unresolved question is if the problem of finding locally EEF allocations parameterized by vertex cover in XP.

%%
%% Bibliography
%%

%% Please use bibtex, 

\bibliography{refs}

\end{document}